\documentclass[twoside,leqno,twocolumn]{article}
\usepackage{graphicx}
\usepackage{amsmath}
\usepackage{amssymb}
 \usepackage{amsthm}
\usepackage{algorithmic}
\usepackage{todonotes}
\usepackage{url}

\newtheorem{Definition}{Definition}[section]
\newtheorem{Theorem}{Theorem}[section]
\newtheorem{Lemma}{Lemma}[section]
\newtheorem{Problem}{Problem}[section]

\newcommand{\I}{\operatorname{I}}
\newcommand{\env}{\operatorname{env}}
\newcommand{\test}{\operatorname{test}}
\newcommand{\train}{\operatorname{train}}

\newcommand{\squishlist}{
  \begin{list}{$$}
  { \setlength{\itemsep}{0pt}
     \setlength{\parsep}{3pt}
     \setlength{\topsep}{3pt}
     \setlength{\partopsep}{0pt}
     \setlength{\leftmargin}{1.5em}
     \setlength{\labelwidth}{1em}
     \setlength{\labelsep}{0.5em} } }

\newcommand{\squishend}{
  \end{list}  }

\begin{document}

\title{Multivariate Confidence Intervals\thanks{
This work was supported by Academy of Finland (decision 288814) and Tekes (Revolution of Knowledge Work).
\newline
\indent This document is an extended version of~\cite{sdm} to appear in the Proceedings of the 2017 SIAM International Conference on Data Mining (SDM 2017).
\newline
\indent Implementation of the algorithms (Sec. \ref{sec:alg}) and code to reproduce Figs. 1 and 5 are provided at {\tt https://github.com\ /jutako/multivariate-ci.git}
\newline
\indent }}

\author{Jussi Korpela\thanks{Finnish Institute of Occupational Health, Helsinki, Finland} 
\and Emilia Oikarinen$^\dag$  \and Kai Puolam{\"a}ki$^\dag$  \and Antti Ukkonen$^\dag$ 
}

\date{}

\maketitle
\begin{abstract} \small\baselineskip=9pt
  Confidence intervals are a popular way to visualize and analyze
  data distributions. Unlike p-values, they can convey information
  both about
  statistical significance as well as effect size. However, very
  little work exists on applying confidence intervals to multivariate
  data. In this paper we define confidence
  intervals for multivariate data that extend the one-dimensional
  definition in a natural way. In our definition every variable is
  associated with its own confidence interval as usual, but a data
  vector can be outside of a few of these, and still be considered to
  be within the confidence area.  We analyze the problem and show that
  the resulting confidence areas retain the good qualities of their
  one-dimensional counterparts: they are informative and easy to
  interpret. Furthermore, we show that the problem of finding
  multivariate confidence intervals is hard, but provide efficient
  approximate algorithms to solve the problem.

\end{abstract}

\section{Introduction}
Confidence intervals are a natural and commonly used way to summarize
a distribution over real numbers.
In informal terms,
a confidence interval is a concise way to
express what values a given sample mostly contains.
They are used widely, e.g., to
denote ranges of data, specify accuracies of parameter estimates, or in
Bayesian settings to describe the posterior distribution.
A confidence interval is given
by two numbers, the lower and upper bound, and
parametrized by the percentage of probability mass that lies within
the bounds.
They are easy to interpret,
because they can be represented
visually together with the data,
and convey information both about
the location as well as the variance of a sample.

There is a plethora of work on how to estimate the
confidence interval of a distribution
based on a finite-sized sample from that distribution, see
\cite{Hyndman96} for a summary.
However, most of these approaches focus on describing
a single univariate distribution over real numbers.
Also, the precise definition of a confidence interval
varies slightly across disciplines and application domains.

In this paper we focus on \emph{confidence areas}: a generalization of 
univariate confidence intervals to multivariate data.
All our intervals and areas are such that they {\em describe ranges of data and are of minimal width}.
In other words, they contain
a given fraction of data within their bounds and 
are as narrow as possible.
By choosing a confidence area with minimal size we essentially locate
the densest part of the distribution.
Such confidence areas are
particularly effective for
visually showing trends, patterns, and outliers.

Considering the
usefulness of confidence intervals, it is surprising how little work
exists on applying confidence intervals on multivariate data
\cite{KorpelaPG14}.
In multivariate statistics {\em confidence regions} are a commonly
used approach,
see e.g.,~\cite{Guilbaud08},
but these methods usually require making assumptions about the underlying distribution.
Moreover, unlike confidence areas, most
conference regions cannot be described simply with an upper and a lower bound,
e.g., confidence regions for multivariate Gaussian data are ellipsoids.
Thus, these approaches do not fully capture 
two useful characteristics of one-dimensional confidence intervals:
a) easy interpretability and b) lack of assumptions about the data.

\begin{figure*}[t!]
  \begin{center}
  	\includegraphics[width=0.45\textwidth]{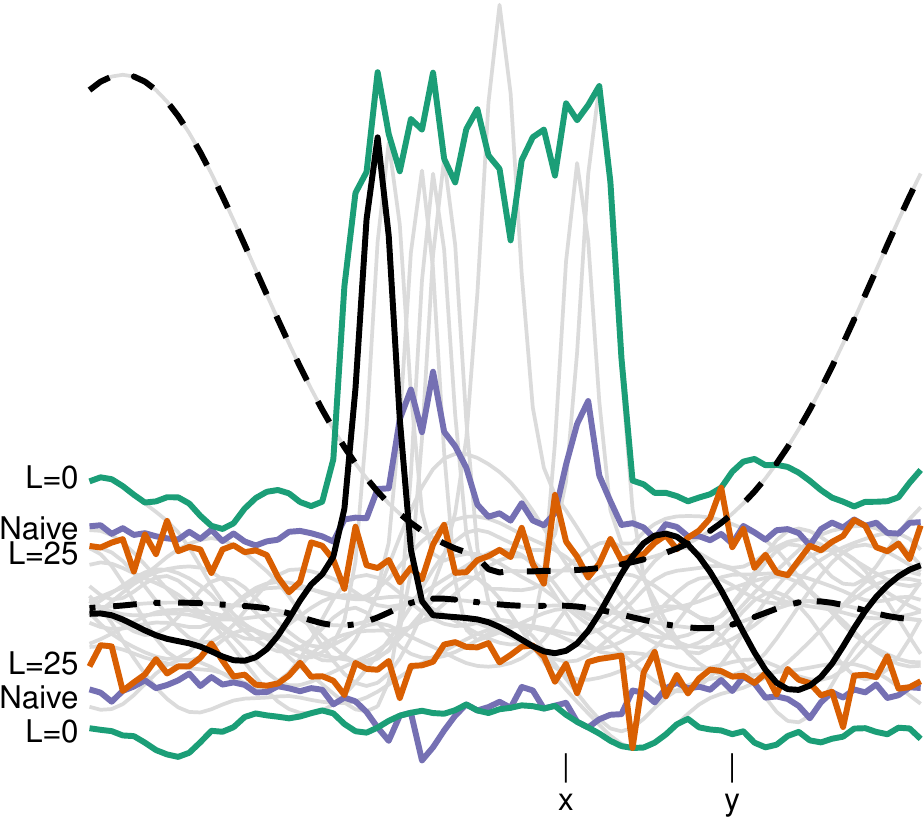}
       \hspace{1em}
        \includegraphics[width=0.45\textwidth]{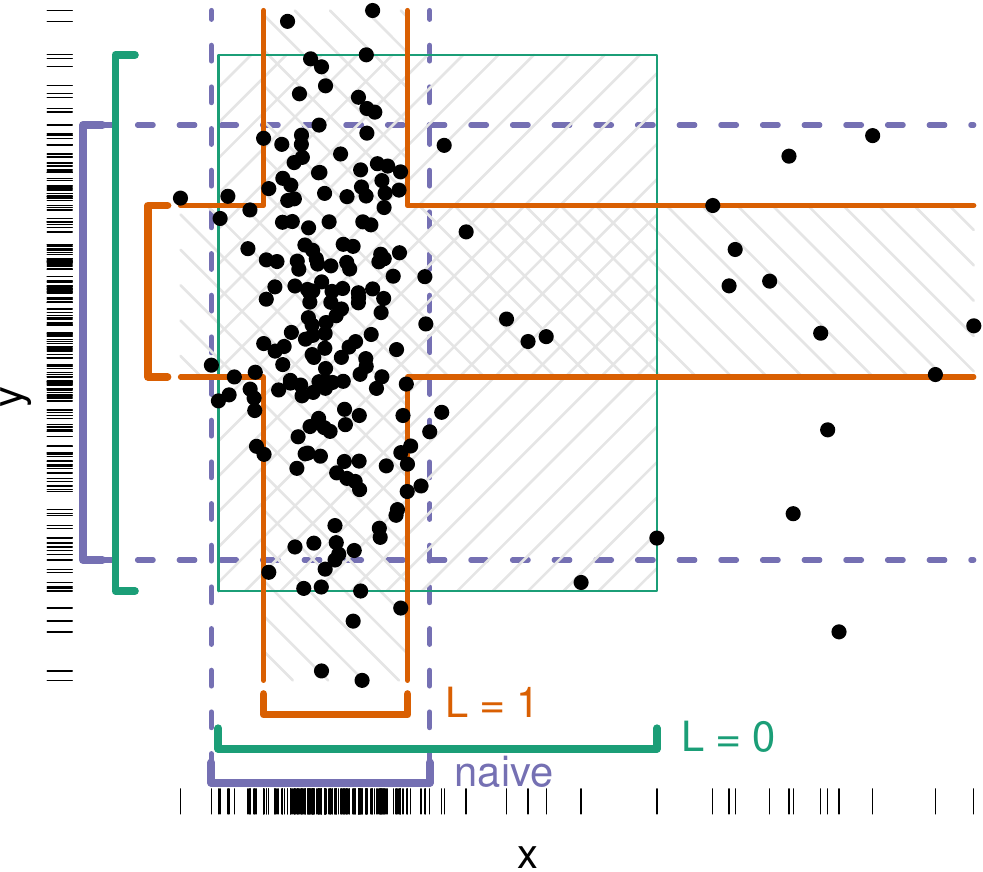}
        \caption{\label{fig:toy} Examples 1 and 2. Please see text for details. LEFT: Local anomalies in time series (solid black line) are easier to spot when computing the confidence area using the proposed method (orange lines) as opposed to existing approaches (green lines). RIGHT: Our approach results in a confidence area (orange ``cross'') that is a better representation of the underlying distribution than existing approaches (green rectangle).}
  \end{center}
\end{figure*}

The simplest approach to construct multivariate confidence intervals
is to compute one-dimensional intervals separately for every variable.
While this naive approach satisfies conditions a) and b) above,
it is easy to see how it fails in general.
Assume, e.g., that we have 10~independently distributed
variables, and have computed for each variable a 90\% confidence interval.
This means that when considering every variable individually,
only 10\% of the distribution lies outside of the respective interval,
as desired.
But taken together,
the probability that an
observation is outside at least one of the confidence intervals is
as much as $1-0.9^{10}=65$\%.
This probability, however, depends strongly on the correlations between the
variables. If the variables were perfectly correlated with a
correlation coefficient of $\pm 1$, the probability of an
observation being outside all confidence intervals would again be 10\%.
In general the correlation structure of the data affects in a
strong and hard-to-predict way on how the naive confidence intervals
should be interpreted in a multivariate setting. 

Ideally a multivariate confidence area 
should retain the simplicity of a one-dimensional confidence interval.
It should
be {\em representable by upper and lower bounds for each variable} and
the semantics should be the same: {\em each observation is either inside
or outside} the confidence area,
and most observations of a sample should be inside.
Particularly important
is specifying when an observation in fact is within the confidence area,
as we will discuss next.

Confidence areas for time series data have been
defined~\cite{Kolsrud2007, KorpelaPG14} in terms of
the {\em minimum width envelope} (MWE) problem:
a time series is within a confidence area if it is
within the confidence interval of {\em every variable}. While this
has desirable properties, it can, however, result in
very conservative confidence areas
if there are local deviations from what constitutes
``normal'' behavior.
The definition in \cite{KorpelaPG14}
is in fact too strict by requiring an observation to
be contained in all variable-specific intervals.

Thus, here
we propose an alternative way to define the confidence
area: {\em a data vector is within the confidence area if it is
outside the variable-specific confidence intervals at most $l$
times}, where $l$ is a given integer.
This formulation preserves easy interpretability:
the user knows that any observation within the confidence area
is guaranteed to violate at most $l$ variable-specific confidence intervals.
In the special case $l=0$,
this new definition coincides with the MWE
problem~\cite{KorpelaPG14}.
The following examples illustrate further properties and uses
of the new definition.

\subsection*{Example 1: local anomalies in time series.}
The left panel of Fig.~\ref{fig:toy}
presents a number of time series over $m = 80$ time points,
shown in light gray,
together with three different types of 90\% confidence intervals,
shown in green, purple and orange, respectively.
In this example,
``normal'' behavior of the time series
is exemplified by the
black dash-dotted curve
that exhibits no noteworthy fluctuation over time.
Two other types of time series are shown also in black:
a clear outlier (dashed black line),
and a time series that exhibits normal behavior most of the time,
but strongly deviates from this for a brief moment (solid black line).

In the situation shown,
we would like the confidence interval
to only show what constitutes ``normal'' behavior,
i.e., not be affected by strong local fluctuations or outliers.
Such artifacts can be caused, e.g., by measurement problems,
or other types of data corruption.
Alternatively, such behavior can also reflect some interesting
anomaly in the data generating process,
and this should not be characterized as ``normal'' by the confidence interval.
In Fig.~\ref{fig:toy} (left)
the confidence area based on MWE \cite{KorpelaPG14},
is shown by the green lines; recall
the MWE interval corresponds to setting $l = 0$.
While it is unaffected by the outlier,
it strongly reacts to local fluctuations.
The naive confidence area,
i.e., one where we have computed
confidence intervals for every time point individually,
is shown in purple.
It is also affected by local peaks,
albeit less than the $l = 0$ confidence area.
Finally,
the proposed method is shown in orange.
The area is computed using $l = 25$, i.e.,
a time series is still within the confidence area
as long as it lies outside
in {\em at most 25 time points}.
This variant focuses on what we would expect to be normal behavior in this case.
Our new definition of a confidence area is thus
nicely suited for {\em finding local anomalies in time series data}.

\subsection*{Example 2: representing data distributions.}
On the other hand,
the right panel of Fig.~\ref{fig:toy}
shows an example where we focus only on the time points $x$ and $y$
as indicated in the left panel.
Time point $x$ resides in the region with a strong local fluctuation,
while at $y$ there are no anomalies.
According to our definition,
an observation,
in this case an $(x,y)$ point,
is within the confidence area if it is outside the variable-specific
confidence intervals at most $l$ times.
We have computed two confidence areas using our method,
one with $l=0$ (green), and another with $l=1$ (orange),
as well as the naive confidence intervals (purple).

For $l=0$, we
obtain the green rectangle in Fig.~\ref{fig:toy} (right panel).
The variable specific
confidence intervals have been chosen so that the green rectangle
contains 90\% the data and the sum of the widths of the confidence
intervals (sides of the rectangle) is minimized.
For $l=1$, we
obtain the orange ``cross'' shaped area.
The cross shape follows from allowing an observation
to exceed the variable specific confidence interval in $l=1$ dimensions.
Again,
the orange cross contains 90\% of all observations, and has
been chosen by greedily minimizing the 
sum of the lengths of the respective variable-specific
confidence intervals.
It is easy to see that with $l = 0$,
i.e., when using the MWE method \cite{KorpelaPG14},
the observations do not occur evenly in
the resulting confidence area (green rectangle).
Indeed, the top right and bottom right parts of the rectangle are mostly empty.
In contrast,
with $l=1$,
the orange cross shaped confidence area
is a much better description of the
underlying data distribution,
as there are no obvious ``empty'' parts.
Our novel confidence area is thus
{\em a better representation of the underlying data distribution}
than the existing approaches.

\subsection*{Contributions}
The basic problem definition we study in this paper is
straightforward: for $m$-dimensional data and the parameters
$\alpha\in[0,1]$ and integer $l$, find a confidence interval for each
of the variables so that the $1-\alpha$ fraction of the observations
lie within the confidence area, defined so that
the sum of the length of the
intervals is minimized,
and an observation
can break at most $l$ of the variable-specific
confidence intervals.
We make the following contributions in this paper:

\begin{itemize}
\item We formally define the problem of finding a multivariate
  confidence area, where observations have to satisfy most but not all
  of the variable-specific confidence intervals.
\item We analyze the computational complexity of the problem, and show
  that it is NP-hard, but admits an approximation algorithm based on a
  linear programming relaxation.
\item We propose two algorithms, which 
  produce good confidence areas in practice.
\item We conduct 
experiments 
demonstrating various
  aspects of multivariate confidence intervals.
\end{itemize}
The rest of this paper is organized as follows: 
Related work is discussed in Sec.~\ref{sec:related}. 
We define the \textsc{ProbCI} and \textsc{CombCI} problems formally in
Sec.~\ref{sec:def}. 
In Sec.~\ref{sec:theory} we study
theoretical properties of the proposed confidence areas, as well as
study problem complexity. Sec.~\ref{sec:alg} describes algorithms
used in experiments in Sec.~\ref{sect:experiments}.  Finally,
Sec.~\ref{sec:concl} concludes this work.

\section{Related work}
\label{sec:related}

Confidence intervals have recently gained more popularity,
as  
they convey information both of statistical significance of
the result {\em as well as the effect size}. In contrast, p-values
give information only of the statistical significance: it is possible
to have statistically significant results that are meaningless in
practice due to the small effect size. The problem 
has been long and
acutely recognized, e.g., in medical research \cite{Gardner86}. Some
psychology journals have recently  
banned the use of p-values
\cite{Woolston15,Trafimow15}. The proposed solution is not to report
p-values at all, but use confidence intervals instead \cite{Nuzzo14}.

Simultaneous confidence intervals for time series data have been
proposed \cite{Kolsrud2007, KorpelaPG14}. These correspond to the 
confidence areas in this paper when $l=0$.
The novelty here 
generalization of the confidence areas to allow outlying dimensions
($l>0$), similarly to \cite{Wolf2015}, and the related theoretical and
algorithmic contributions, allowing for narrower confidence intervals
and in some cases more readily interpretable results. Simultaneous
confidence intervals with $l=0$ were 
in \cite[p. 154]{Davidson97}, and studied \cite{Mandel08} by using the
most extreme value within a data vector as a ranking
criterion. Another examples of $l=0$ type procedures include
\cite{Liu07,Schussler16}. In the field of information visualization and the
visualization of time series confidence areas have been used
extensively; see \cite{Aigner11} for a review.  An interesting
approach is to construct simultaneous confidence regions by inverting
statistical multiple hypothesis testing methods, see
e.g.,~\cite{Guilbaud08}.

The approach presented in this paper
allows some dimensions of an observation to be partially outside the
confidence area. This is in the same spirit---but not equivalent---to
false discovery rate (FDR) in multiple hypothesis testing, which also
allows a controlled fraction of positives to be ``false positives''.
In comparison, the approach in \cite{KorpelaPG14} is more akin to
family-wise error rate (FWER) that controls the probability of at
least one false discovery.

\section{Problem definition}
\label{sec:def}

A {\em data vector} $x$ is a vector in ${\mathbb{R}}^m$ and $x(j)$
denotes the value of $x$ in $j$th position.  Let matrix
$X\in{\mathbb{R}}^{n\times m}$ be a dataset of $n$ data vectors
$x_1,\ldots,x_n$, i.e. rows of $X$.  We start by defining the
confidence area, its size, and the envelope for a dataset $X$.

\begin{Definition}
Given $X\in{\mathbb{R}}^{n\times m}$, a {\em confidence area} for $X$
is a doublet of vectors $(x_l,x_u)$, $x_l,x_u\in {\mathbb{R}}^m$
composed of lower and upper bounds satisfying $x_l(j)\le x_u(j)$ for
all $j$, respectively.  The {\em size} of the confidence area is
$A=\sum_{j=1}^m{w(j)}$, where $w(j)=x_u(j)-x_l(j)$ is the {\em width}
of the confidence interval w.r.t.\ $j$th position.  The {\em envelope}
of $X$ is a confidence area denoted by $\env(X)=(x_l,x_u)$, where
$x_l(j)=\min_{i=1}^n{x_i(j)}$ and $x_u(j)=\max_{i=1}^n{x_i(j)}$.
\end{Definition}
We define the error of the confidence area as the number of outlying
data vectors as follows.
\begin{Definition}
Let $x$ be a data vector in $\mathbb{R}^m$ and $(x_l,x_u)$ a
confidence area.  The {\em error} of $x$ given the confidence area is
defined as
\begin{equation}
\nonumber
V(x \mid x_u, x_l) = \sum_{j=1}^m \I \left[x(j)<x_l(j)\vee x_u(j)<x(j)
  \right].
\end{equation}
\end{Definition}
The indicator function $\I[\Box]$ is unity if the condition
$\Box$ is satisfied and zero otherwise.

The main problem 
in this paper is as follows.
\begin{Problem}[\textsc{ProbCI}]
\label{prob:pci}
Given $\alpha \in [0,1]$, an integer~$l$, and a distribution $F$ over
$\mathbb{R}^m$, find a confidence area $(x_l,x_u)$ that minimizes
$\sum_{j=1}^m w(j)$ subject to 
constraint
\begin{equation}
  \Pr_{x \sim F}\left[ V(x \mid x_u, x_l) \le l \right] \ge 1-\alpha.
  \label{eq:p1}
\end{equation}
\end{Problem}
For this problem definition to make sense, the variables or at
least their scales must be comparable. Otherwise variables
with high variance will dominate the confidence areas. Therefore, some
thinking and a suitable preprocessing, such as normalization of
variables, should be applied before solving for the confidence area
and interpreting it.
 
The combinatorial problem 
is defined 
as follows.
\begin{Problem}[\textsc{CombCI}]
\label{prob:cci}
Given integers $k$ and $l$, and $n$ vectors $x_1, \ldots, x_n$ in $\mathbb{R}^m$, find a confidence area $(x_l,x_u)$ by
minimizing $\sum_{j=1}^m w(j)$ subject to constraint
\begin{equation}
\label{eq:cciconstraint}
\sum_{i=1}^n \I \left[ V(x_i \mid x_u, x_l ) \le l \right] \ge n-k.
\end{equation}
\end{Problem}
Any confidence area satisfying \eqref{eq:cciconstraint} is called a
\emph{$(k,l)$-confidence area}.  In the special case with $l=0$,
Problems~\ref{prob:pci} and \ref{prob:cci} coincide with the minimum
width envelope problem from \cite{KorpelaPG14}.  The problem
definition with non-vanishing $l$ is novel to best of our knowledge.
The relation of Problems~\ref{prob:pci}~and~\ref{prob:cci} is as
follows.
\begin{Definition}
Problems \ref{prob:pci} and \ref{prob:cci} {\em match} for a given
data from distribution $F$ and parameters $\alpha$, $k$, and $l$, if a
solution of Problem \ref{prob:cci} satisfies Eq. \eqref{eq:p1} with
equality for a previously unseen test data from distribution $F$.
\label{def:match}
\end{Definition}
We can solve Problem \ref{prob:pci} by solving Problem \ref{prob:cci}
with different values of $k$ to find a $k$ that matches the given
$\alpha$, as done in Sec. \ref{sec:match} or \cite{KorpelaPG14}.

\section{Theoretical observations}
\label{sec:theory}

\subsection{Confidence areas for uniform data}
\label{sec:uniform}

It is instructive to consider the behavior of Problem \ref{prob:pci}
with the uniform distribution $F_{unif}=U(0,1)^m$. 
We
show that a solution may contain very narrow confidence intervals and
discuss the required number of data vectors to estimate a confidence
area with a desired level of $\alpha$.

Consider Problem \ref{prob:pci} in a simple case of two-dimensional
data with $m=2$, $l=1$, and $\alpha=0.1$. In this case an optimal
solution to Problem \ref{prob:pci} is given by confidence intervals 
with $w(1)=0.9$ and $w(2)=0$,  
resulting in the
size of
confidence area of $\sum_{j=1}^2{w(j)}=0.9$. A solution with
confidence intervals of equal width, i.e., $w(1)=w(2)=0.68$ would  
lead to substantially larger area of $1.37$. As this simple
example demonstrates, if data is unstructured or if some of the
variables have unusually high variance, we may obtain solutions where
some of the confidence intervals are very narrow. In the case of
uniform distribution the choice of variables with zero width
confidence intervals is arbitrary: e.g., in the example above we could
as well have chosen $w(1)=0$ and $w(2)=0.9$. Such narrow intervals are
not misleading, because they are easy to spot: for a narrow
interval---such as the one discussed in this paragraph---only a small
fraction of the values of the variable lie within it. In real data
sets with non-trivial marginal distributions and correlation structure
the confidence intervals are often of roughly similar width.

Next we  
consider the behavior of the
problem at the limit of high-dimensional data.
\begin{Lemma}
  The solution with confidence intervals of equal width
  $w=w_j=x_u(j)-x_l(j)$ corresponds to $\alpha$ of
  \begin{equation}
    \alpha=1-BC(l,m,w),
    \label{eq:bc}
   \end{equation}
  where $BC(l,m,w)=\sum_{j=0}^l \binom{m}{j}(1-w)^jw^{m-j}$ is the cumulative binomial
  distribution.
  \label{lem:eq}
\end{Lemma}
\begin{Lemma}
  If $n$ vectors are sampled from $F_{unif}=U(0,1)^m$ then the
  expected width of the envelope for each variable is
  $w=\frac{n-1}{n+1}$. The probability of more than $l$ variables from a
  vector from $F_{unif}$ being outside the envelope is given by
  Eq. \eqref{eq:bc} with $w=\frac{n-1}{n+1}$.
  \label{lem:n}
\end{Lemma}

One implication of Lemma \ref{lem:n} is that there is a limit to the
practical accuracy that can be reached with a finite number of
samples. The smallest $\alpha$ we can hope to realistically reach is
the one implied by the envelope of the data, unless we make some
distributional assumptions of the shape of the distribution outside
the envelope. Conversely, the above lemmas define a minimum number of
samples needed for uniform data to find the confidence area for a
desired level of $\alpha$.

In the case of $l=0$ to reach the accuracy of $\alpha$,
we have $w^m\approx 1-\alpha$, or $n\approx
-2m/\log{\left(1-\alpha\right)}\approx 2m/\alpha$, where we have
ignored higher order terms in $1/m$ and $\alpha$. For a typical choice
of $\alpha=0.1$ and $m=100$ this would imply that at least $n\approx
2000$ data vectors are needed to estimate the confidence area; the
number of data vectors needed increases linearly with the
dimensionality $m$.

On the other hand, for a given $\beta\in]0,1[$, if we let
    $l=\lfloor\beta m\rfloor$, a solution where the width of the
    envelope is $w\approx 1-\beta$ is asymptotically sufficient when
    the dimensionality $m$ is large, in which case the number of data	
    vectors required is $n\approx 2/\beta$. For a value of $\beta=0.1$
    only $n\approx 20$ data vectors would therefore be needed even for
    large $m$. As hinted by the toy example in Fig.~\ref{fig:toy}, a
    non-vanishing parameter $l$ leads at least in this example to
    substantially narrower confidence intervals and, hence, make it
    possible to compute the confidence intervals with smaller data
    sets.

\subsection{Complexity and approximability}

We continue by briefly addressing the computational properties of the
\textsc{CombCI} problem.
The proofs are provided in the appendix.

\begin{Theorem}
\label{thr:np}
\textsc{CombCI} is NP-hard for all $k$.
\end{Theorem}
For $k>0$ the result directly follows from \cite[Theorem 3]{KorpelaPG14}, and
for $k=0$ and $l>0$ a reduction from \textsc{Vertex-Cover} can be used.

Now, there exists a {\em polynomial time approximation
 algorithm} for solving a variant of \textsc{CombCI} in the special
case of $k=0$. In particular, we consider here a {\em one-sided}
confidence area, defined only as the {\em upper bound} $x_u$, the
lower bound $x_l$ is assumed to be fixed, e.g., at zero, or any other
suitable value. 
This complements the earlier result that
the {\em complement} of the objective function of \textsc{CombCI} is
hard to approximate for $k>0$ and $l=0$ \cite[Theorem 3]{KorpelaPG14}.
\begin{Theorem}
\label{thr:app}
There is a polynomial time $(l+1)$ approximation algorithm for the
one-sided \textsc{CombCI} problem with $k=0$.
\end{Theorem}
The proof uses a linear programming relaxation
of an {\em integer linear program} corresponding to  the $k=0$
variant of the one-sided \textsc{CombCI} 
and the approximation ratio obtained from the solution given by the relaxation. 

Finding a bound that does not depend on $l$ is an interesting open
question, as well as extending the proof to the two-sided
\textsc{CombCI}.  It is unlikely that the problem admits a better
approximation bound than $2$, since this would immediately result in a
new approximation algorithm for the \textsc{Vertex-Cover} problem.
This is because in the proof of Theorem \ref{thr:app} 
we describe a simple
reduction from \textsc{Vertex-Cover} to \textsc{CombCI} with $l=1$ and
$k=0$.  This reduction preserves approximability, as it maps the
\textsc{Vertex-Cover} instance to an instance of our problem in a
straightforward manner.  For \textsc{Vertex-Cover} it is known that
approximation ratios below $2$ are hard to obtain in the general case.
Indeed, the best known bound is equal to $2 -
\Theta(1/\sqrt{\log|V|})$ \cite{Karakostas2009}.

\section{Algorithms}
\label{sec:alg}

We present two algorithms for  
$(k,l)$-confidence~areas.

\subsection{Minimum intervals ({\sc mi})}
Our first method is based on {\em minimum intervals}. 
A standard approach to define a confidence interval for univariate data is to 
consider the minimum length interval that contains $(1-a)$\% of the observations.
This can be generalized for multivariate data by
treating each variable independently, i.e.,
for a given $a$,
we set $x_l(j)$ and $x_u(j)$
equal to lower and upper limit of
such a minimum length interval
for every $j$. 
The {\sc mi} algorithm solves the \textsc{CombCI} 
problem for given $k$ and $l$ by adjusting the parameter $a$ such that the resulting $(x_u, x_l)$ satisfies the constraint 
in Eq.~\eqref{eq:cciconstraint} in the training data set. 
The resulting confidence area may have a larger size than the optimal 
solution, since all variables use the same $a$.
Note that the {\sc mi} solution differs from the naive 
solution mentioned in the introduction because the naive solution does not adjust $a$ but 
simply sets it to $a=k/n$.
The time complexity of {\sc mi} with binomial search for correct $a$
is $O(mn\log{k}\log{n})$.

\begin{figure*}[t]
\centering
\includegraphics[width=0.999\textwidth]{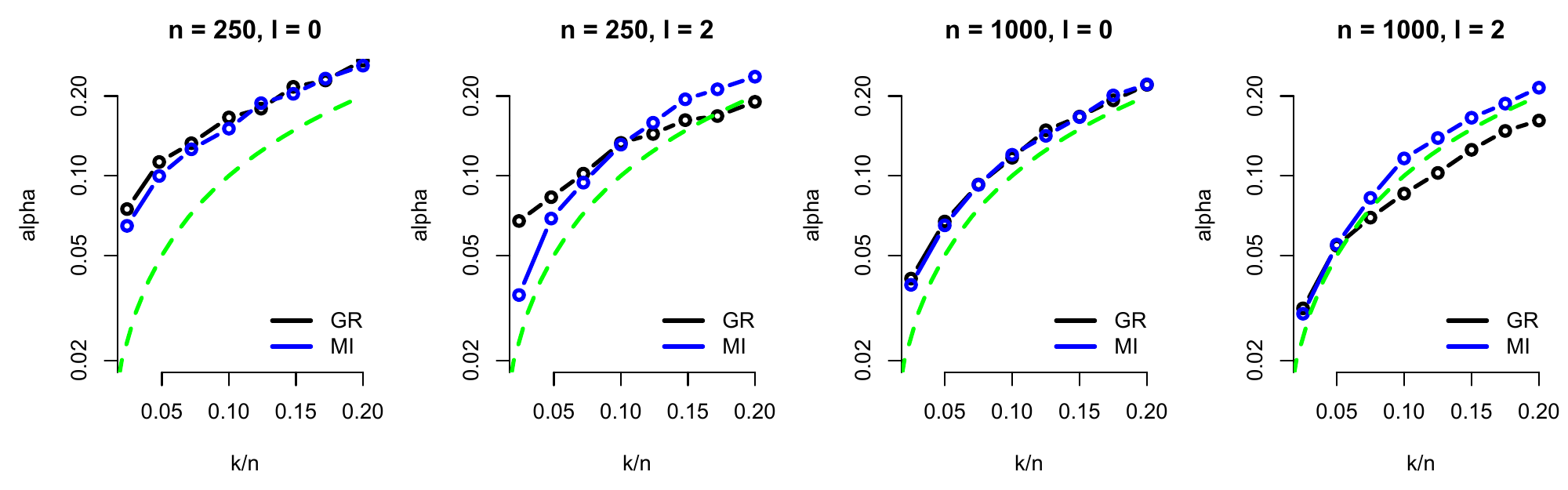}
\caption{Comparison of $k/n$ used when fitting the
  confidence interval on training data and the observed level of $\alpha$ in a separate test data. 
  Both data are normal with $m=10$, $n_{test} = 1000$, $n_{train} = \{250,1000\}$
  and $l = \{0,2\}$. The dashed green line indicates $k/n = \alpha$. 
  Shown are the averages over 25
  independent trials over different randomly generated training and
  test instances. Note the
  log-scale on the vertical axis.}
\label{fig:n_l_comparison}
\end{figure*}

\subsection{Greedy algorithm ({\sc gr})}
Our second method is a greedy algorithm. 
The greedy search could be done either {\em bottom-up} (starting from an empty set of included data vectors and then adding $n-k$ data vectors) or {\em
  top-down} (starting from $n$ data vectors and  by excluding $k$  data vectors). 
Since typically $k$ is smaller than $n-k$ 
we will consider here a top-down greedy algorithm.

The idea of the greedy algorithm 
is to start from the envelope of whole data
and sequentially find $k$ vectors to exclude by selecting at each
iteration the vector whose removal reduces the envelope the largest
amount.
In order to find the envelope wrt.\ the relaxed condition allowing $l$
positions from each vector to be outside, at each iteration the set of
included data points needs to be (re)computed. 
This is done
by implementing a greedy algorithm solving the
 \textsc{CombCI} problem for $k=0$. Here one removes individual points that result in maximal decrease in the envelope size so that at most $l$ points from each data vector are be removed, thus obtaining the  envelope wrt.\ the $l$ criterion.
 After this envelope has been computed, the data vector whose
exclusion yields a maximal decrease in the size of the confidence area
is excluded.  
For this, a slightly modified variant of the
greedy MWE algorithm from
\cite{KorpelaPG14} with $k=1$ is used.
After $k$ data vectors have been excluded, the final set of points 
included in the envelope is returned.
The time complexity of {\sc gr}  is $O(mkn\log{n})$.

\section{Empirical evaluation}
\label{sect:experiments}
We present here empirical evaluation of the algorithms. 
In the following {\sc mi} and {\sc gr} refer to the Minimum
intervals and Greedy algorithm, respectively.

\subsection{Datasets}
We make experiments using various datasets that reflect different
properties  
of interest from the point of view of fitting
confidence areas. In particular, we want to study the effects of
correlation (autocorrelation in the case of time-series or regression
curves), number of variables, and distributional qualities.

\squishlist
\item {\bf Artificial data.}  We use artificial multivariate
  (i.i.d.~variables) datasets sampled from the uniform distribution
  (in the interval $[0,1]$), the standard normal distribution, and the
  Cauchy distribution (location parameter $0$, scale parameter $\gamma
  = 2$), with varying $n$ and $m$ to study some theoretical properties
  of multivariate confidence areas.

\item {\bf Kernel regression data.}  We use the Ozone and
  South African heart disease (SA heart) datasets (see, e.g.,
  \cite{eslbook})  
  to compute  
  confidence
  areas for bootstrapped kernel regression estimates.  We use a simple
  Nadaraya-Watson kernel regression estimate to produce the vectors
  $X$, and fit confidence intervals to these using our algorithms.  By
  changing the number of bootstrap samples we can vary $n$, and by
  altering the number of points in which the estimate is computed we
  can vary $m$.

\item {\bf Stock market data.}  We obtained daily closing prices of $n
  = 400$ stocks for years 2011--2015 
  ($m = 1258$ trading
  days).  The time-series are normalized to 
  reflect the
  absolute change in stock price with respect to the average price of
  the first five trading days in the data. The data is shown in
  Fig.~\ref{fig:stock2}.
\squishend

\vspace{-2ex}
\subsection{Finding the correct $\mathbf{k}$}
\label{sec:match}
Our algorithms both solve the {\sc CombCI} problem
(Problem~\ref{prob:cci}) in which we must specify the number of vectors
$k$ that are allowed to violate 
the confidence area.  To obtain a matching $\alpha$ in
{\sc ProbCI} (Definition \ref{def:match}) we must choose the parameter
$k$ appropriately. 
We study here how this can be achieved.

\begin{figure}[t]
  \includegraphics[width=\columnwidth]{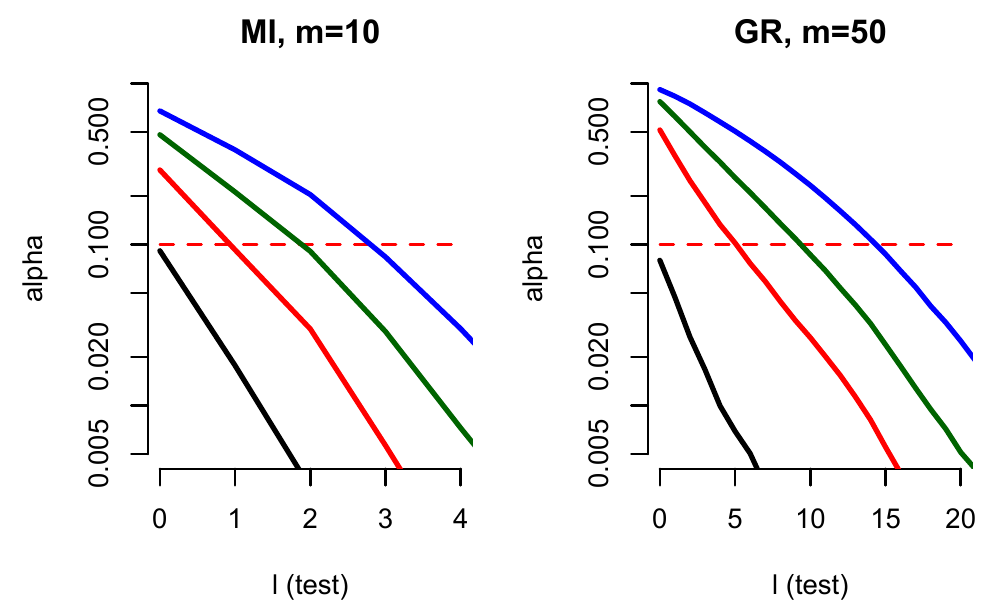}
\caption{Effect of the value of $l_{\test}$ on the observed $\alpha$
  for different values of $l_{\train}$. On the left $l_{\train}$ is
  equal to $0$ (black), $1$ (red), $2$ (green) and $3$ (blue), while
  on the right $l_{\train}$ is $0$ (black), $5$ (red), $10$ (green),
  and $15$ (blue).  In every case the confidence area was trained to
  have $\alpha = 0.1$ for the given $l_{\train}$.}
\label{fig:l_vs_alpha}
\end{figure}

Fig.~\ref{fig:n_l_comparison} shows $\alpha$ as a function of $k/n$
in data with $m = 10$ independent normally distributed variables for
four combinations of $n$ and $l$.  The dashed green line shows the
situation with $k/n = \alpha$.  (This is a desirable property as it
means fine-tuning $k/n$ is not necessary to reach some particular
value of $\alpha$.)  We observe from Fig.~\ref{fig:n_l_comparison}
that when the data is relatively small ($n=250$), for a given $k/n$
both {\sc gr} as well as {\sc mi} tend to find confidence areas that
are somewhat too narrow for the test data leading to $\alpha > k/n$.
To obtain some particular $\alpha$, we must thus set $k/n$ to a lower
value.  For example, with $n=250$ and $l=0$, to have $\alpha = 0.1$ we
must let $k/n = 0.05$.  As the number of training examples is
increased to $n=1000$, we find that both algorithms are closer to the
ideal situation of $k/n = \alpha$.  Interestingly, this happens also
when when $l$ increases from $l=0$ to $l=2$.  The relaxed notion of
confidence area we introduce in this paper is thus somewhat less prone
to overfitting issues when compared against the basic confidence areas
with $l=0$ of \cite{KorpelaPG14}.  On the other hand, for $n=1000$ we
also observe how {\sc gr} starts to ``underfit'' as $k/n$ increases,
meaning that we have $\alpha < k/n$.  Errors in this direction,
however, simply mean that the resulting confidence area is
conservative and will satisfy the given $k/n$ by a margin.

The dependency between $k/n$ and $\alpha$ and other observations made
above are qualitatively identical for uniform and Cauchy distributed
data (not shown).

\subsection{Effect of $l$ on $\alpha$ in test data}
Note that the value of $\alpha$ that we compute for a given confidence
area also depends on the value of $l$ used when {\em evaluating}.  We
continue by showing how $\alpha$ depends on the value of $l$ 
used when evaluating the confidence area on test data.  In this
experiment we train confidence areas using the Ozone data (with
$n=500$ and $m=10$ or $m=50$) and adjust $k/n$ so that we have $\alpha
= 0.1$ for a given value of $l_{\train}$, where $l_{\train} \in \{ 0,
0.1m, 0.2m, 0.3m\}$ is the parameter given to our algorithms ({\sc
  mi}, {\sc gr}) to solve Prob.~\ref{prob:cci}.  Then we estimate
$\alpha$ for all $l_{\test} \in \{1, \ldots, m\}$ using
Eq.~\eqref{eq:p1} and previously unseen test data set.

Results are shown in Fig.~\ref{fig:l_vs_alpha}.  We find that in every
case the line for a given value of $l_{\train}$ intersects the thin
dashed (red) line (indicating $\alpha = 0.1$) at the correct value
$l_{\test} = l_{\train}$.  More importantly, however, we also find
that $l_{\test}$ has a very strong effect on the observed $\alpha$.
This means that if we relax the condition under which an example still
belongs in the confidence area (i.e., increase $l_{\test}$), $\alpha$
drops at a very fast rate, meaning that a confidence area trained for
a particular value of $l_{\train}$, will be rather conservative when
evaluated using $l_{\test} > l_{\train}$.  Obviously the converse
holds  
as well, 
i.e., when $l_{\test} < l_{\train}$, the
confidence area will become much too narrow very fast.  This implies
that $l_{\train}$ should be set conservatively (to a low value) when
it is important to control the false positive rate, e.g., when the
``true'' number of noisy dimensions is unknown.

\subsection{Algorithm comparison}
Next we briefly compare the {\sc mi} and {\sc gr} algorithms in terms
of the confidence area size (given as $A/m$, where $m$ is the number
of variables) and running time $t$ (in seconds).  Results for artificial
data sets as well as the two regression model data are shown in
Table~\ref{table:results_table}.  The confidence level (in test data)
was set to $\alpha = 0.1$ in every case, and all training data had $n
= 500$ examples.  All numbers are averages of 25 trials.  We can
observe, that {\sc mi} tends to produce confidence areas that are
consistently somewhat smaller than those found by {\sc gr}.  Also,
{\sc mi} is substantially faster, albeit our implementation of {\sc
  gr} is by no means optimized.  Finally, the $k/n$ column shows the
confidence level that was used when fitting the confidence area to
obtain $\alpha = 0.1$.  Almost without exception, we have $k/n <
\alpha$ for both algorithms, with {\sc mi} usually requiring a
slightly larger $k$ than {\sc gr}.  Also worth noting is that for
extremely skewed distributions, e.g., Cauchy, the confidence area
shrinks rapidly as $l$ is increased from zero.

\begin{table}
\begin{footnotesize}
\begin{tabular}{@{}l@{ }|r@{\ \ }r|r@{\ \ }r@{\ \ }r|r@{\ \ }r@{\ \ }r}
\multicolumn{3}{c|}{} & \multicolumn{3}{c|}{{\sc mi}} & \multicolumn{3}{c}{{\sc gr}} \\
dataset  & $l$ & $m$ & $A/m$ & $t$ & $k/n$ & $A/m$ & $t$ & $k/n$ \\
\hline
unif &  0 &  10 &   1.0 &   0.1 & 0.06 &   1.0 &   0.7 & 0.07\\
unif &  2 &  10 &   0.9 &   0.2 & 0.08 &   0.9 &  14.0 & 0.11\\
unif & 10 & 100 &   0.9 &   0.7 & 0.05 &   0.9 &  32.1 & 0.03\\
unif & 50 & 500 &   0.9 &   2.1 & 0.03 &   0.9 & 220.2 & 0.03\\
\hline
norm &  0 &  10 &   5.1 &   0.1 & 0.07 &   5.3 &   0.8 & 0.07\\
norm &  2 &  10 &   3.2 &   0.2 & 0.09 &   3.2 &  12.0 & 0.09\\
norm & 10 & 100 &   3.6 &   0.8 & 0.06 &   3.6 &  33.1 & 0.03\\
norm & 50 & 500 &   3.5 &   2.1 & 0.03 &   3.5 & 209.7 & 0.03\\
\hline
Cauchy &  0 &  10 & 262.3 &   0.1 & 0.08 & 325.1 &   0.7 & 0.07\\
Cauchy &  2 &  10 &  21.8 &   0.1 & 0.09 &  25.6 &  10.7 & 0.08\\
Cauchy & 10 & 100 &  36.8 &   0.7 & 0.07 &  38.5 &  31.4 & 0.03\\
Cauchy & 50 & 500 &  30.5 &   2.2 & 0.05 &  31.5 & 203.3 & 0.03\\
\hline
ozone &  0 &  10 &   5.1 &   0.1 & 0.08 &   5.1 &   0.8 & 0.07\\
ozone &  2 &  10 &   3.4 &   0.1 & 0.09 &   3.5 &   9.0 & 0.09\\
ozone &  5 &  50 &   4.2 &   0.3 & 0.09 &   4.3 &  11.4 & 0.05\\
ozone & 15 &  50 &   3.0 &   0.3 & 0.09 &   3.1 &  48.6 & 0.06\\
\hline
SA heart &  0 &  10 &   5.3 &   0.1 & 0.06 &   5.4 &   0.7 & 0.06\\
SA heart &  2 &  10 &   3.3 &   0.2 & 0.08 &   3.7 &  15.5 & 0.14\\
SA heart &  5 &  50 &   4.3 &   0.2 & 0.07 &   4.5 &   8.6 & 0.04\\
SA heart & 15 &  50 &   2.9 &   0.3 & 0.08 &   3.3 &  66.7 & 0.06\\
\hline
\end{tabular}

\end{footnotesize}
\caption{Comparison between {\sc mi} and {\sc gr}.}
\label{table:results_table}
\end{table}

\begin{figure}[t]
\centering
\includegraphics[width=0.96\columnwidth]{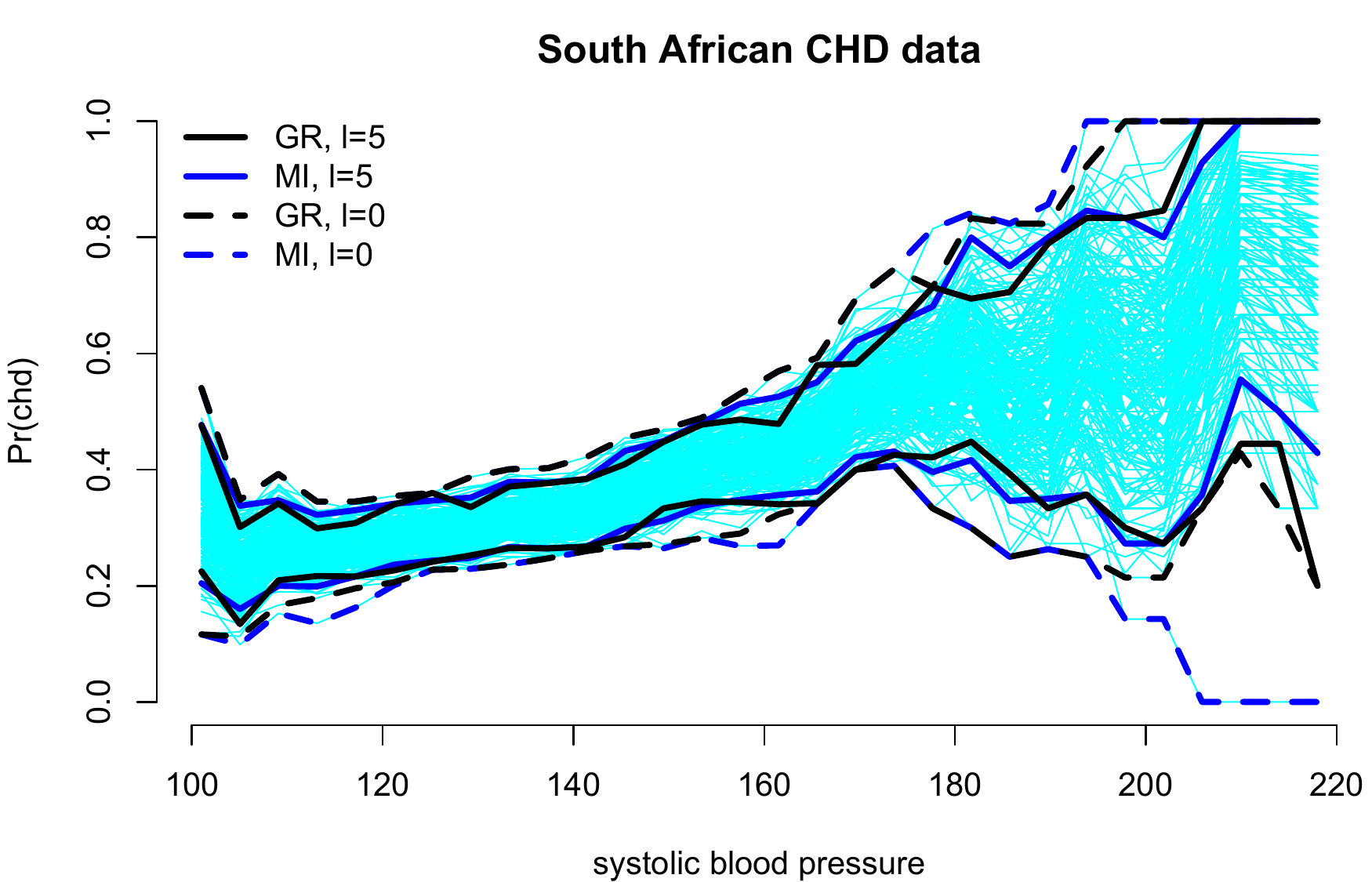}\\
\includegraphics[width=0.96\columnwidth]{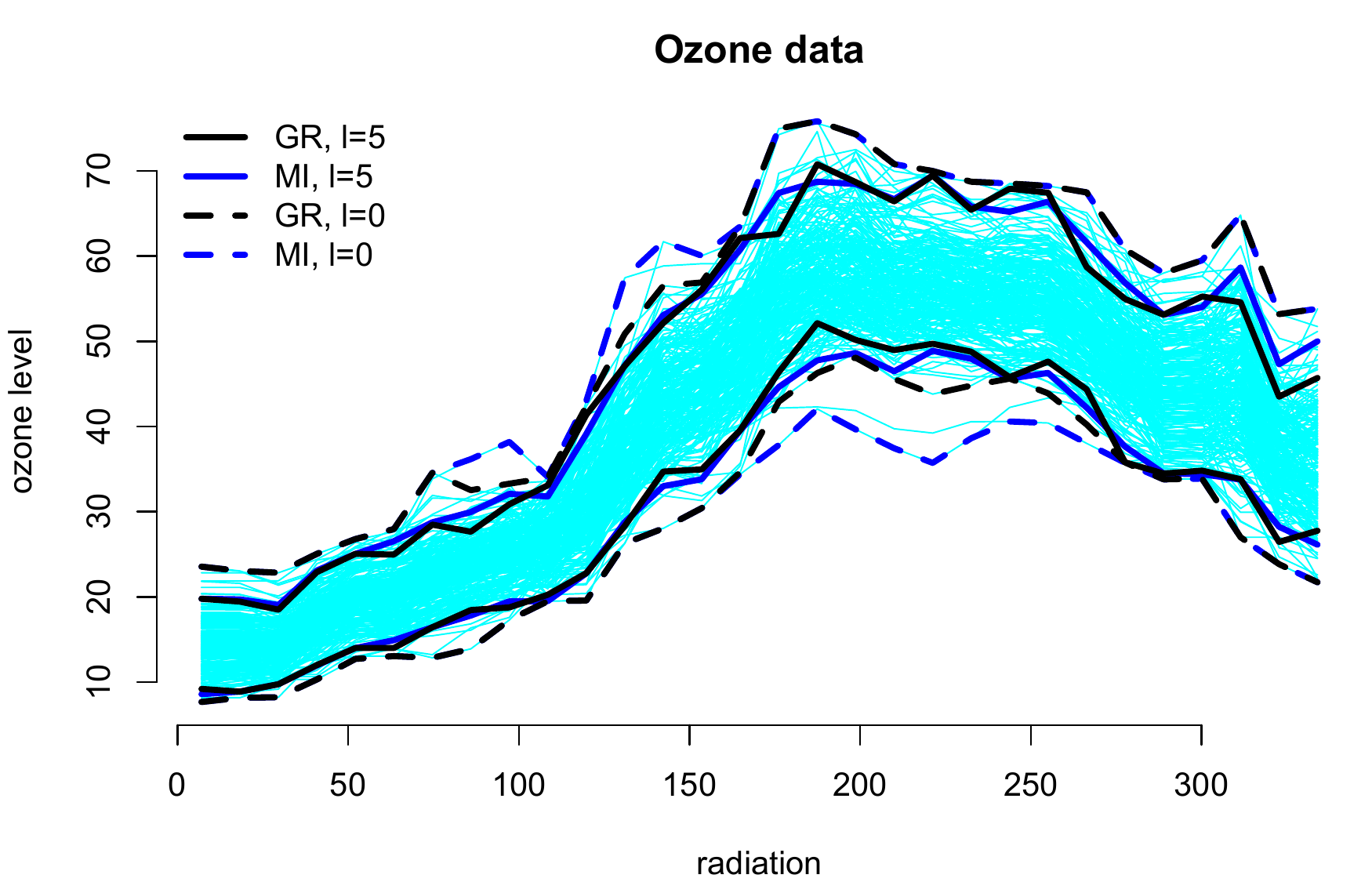}
\caption{\label{fig:kreg_examples}  
  Using {\sc gr} and {\sc mi} on bootstrapped kernel
  regression estimates ($n=250$, $m=30$, $l \in \{0, 5\}$). Top:
  Probability of coronary heart disease as a function of systolic
  blood pressure in the SA heart data.  
  Bottom: Ozone level as a function of observed radiation.}
\end{figure}

\subsection{Application to regression analysis}
We can model the accuracy of an estimate given by a regression model
by resampling the data points, e.g., by the bootstrap method, and then
refitting the model for each of the resampled data sets
\cite{Efron93}. The estimation accuracy or spread of values for a
given independent variable can be readily visualized using confidence
intervals.

Fig.~\ref{fig:kreg_examples} shows examples of different confidence
intervals fitted to bootstrapped kernel regression estimates on two
different datasets using both $l=0$ and $l=5$ ($n=250$ and $m=30$).
In both cases $k$ was adjusted so that $\alpha = 0.1$ in a separate
test data.  We find that qualitatively there is very little difference
between {\sc mi} and {\sc gr} when $l=5$.  For $l=0$, {\sc gr} tends
to produce a somewhat narrower area.  In general, this example
illustrates the effect of $l$ on the resulting confidence area in
visual terms.  By allowing the examples to lie outside the confidence
bounds for a few variables ($l=5$) we obtain substantially narrower
confidence areas that still reflect very well the overall data
distribution.
 
\subsection{Stock market data}
The visualization for the stock market data of Fig.~\ref{fig:stock2}
has been constructed using  
{\sc gr} algorithm with parameters $k=40$
and $l=125$. The figure shows in yellow the  
stocks that are
clearly outliers and among the $k=40$ anomalously behaving observations  
ignored in the construction of the confidence area. The
remaining $n-k=360$ stocks (shown in blue) 
remain within the
confidence area at least $\frac{m-l}{m}=90$\% of the time. However, they are
allowed to deviate from the confidence area 10\% of the
time. Fig.~\ref{fig:stock2} shows several such 
stocks, 
one of them highlighted with red.
By allowing these excursions, the confidence area is smaller and these
potentially interesting deviations are easy to spot.  E.g., the red
line  
represents  
Mellanox Technologies 
and market analyses
from fall 2012 report the stock being overpriced at that time.  The
black line in Fig.~\ref{fig:stock2} 
represents 
Morning\-star,
an example 
staying inside the confidence area.
If we do not allow any deviations outside the confidence intervals,
i.e., we  set $l=0$, then the confidence area will
be larger and such 
deviations might be missed.
\begin{figure*}[t]
  \begin{center}
    \includegraphics[width=0.94\textwidth]{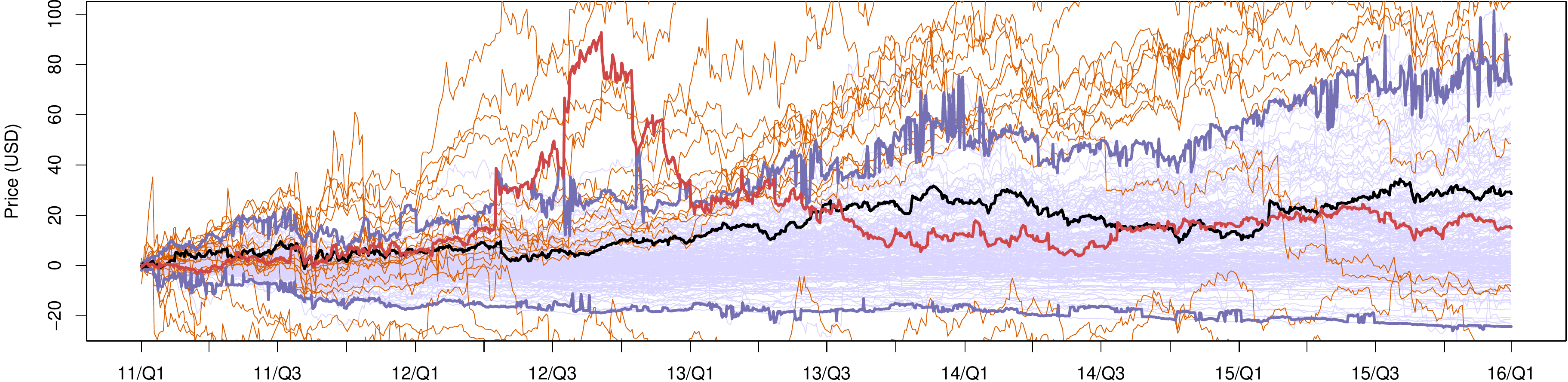}
    \caption{\label{fig:stock2} Visualization of the relative closing
      values of 400 stocks from Jan.~2011 to Dec.~2015 (1258 days) compared
      to the starting days. 
      The confidence area
      with parameters $k=40$ and $l=125$ is shown with blue lines.  An
      example of a valuation of a stock that temporarily deviates (in
      less than $l$ time points) from the confidence area is shown in
      red, and an example of a valuation for a stock observing
      ``normal'' development is shown in black.}
  \end{center}
\end{figure*}

\section{Concluding remarks}
\label{sec:concl}
The versatility of confidence intervals stems from their
simplicity. They are easy to understand and to interpret, and therefore
often used in presentation and interpretation of multivariate
data. The application of confidence intervals to multivariate data is,
however, often done in a naive way, disregarding effects of
multiple variables. This may lead to false interpretations of
results if the user is not being careful.

In this paper we have presented a generalization of confidence
intervals to multivariate data vectors. The generalization is simple
and understandable and does not sacrifice the interpretability of
one-dimensional confidence intervals. 
The confidence areas defined this way behave intuitively and  
offer insight into the data.
The problem of finding confidence areas is computationally hard. We
present two efficient algorithms to solve the problem and show that
even a rather simple approach ({\sc mi}) can produce very good
results.

Confidence
intervals are an intuitive and useful tool for visually presenting and
analyzing data sets, spotting trends, patterns, and outliers. The
advantage of confidence intervals is that they give at the same
time information about both the statistical significance of the result
and size of the effect. 
In this paper,  we have shown several examples
demonstrating the usefulness of multivariate confidence intervals, i.e. confidence areas.

In addition to visual tasks, the confidence intervals can also be used
in automated algorithms as a simple and robust distributional
estimators. As the toy example of Fig. \ref{fig:toy} shows, the
confidence areas with $l>0$ can be a surprisingly good
distributional estimator, if data is sparse, i.e., a majority of
variables is close to the mean and in each data vector only a small
number of variables have significant deviations from the average.

With p-values there are established procedures to deal with 
multiple hypothesis testing. Indeed, a proper
treatment of the multiple comparisons problem is required, e.g., in
scientific reporting. 
Our contribution to the discussion about reporting scientific results is
to point out that it is indeed possible to treat multidimensional data
with confidence intervals in a principled~way.

\bibliographystyle{abbrv}
\begin{small}
\bibliography{fdrbands}
\end{small}
\appendix

\section{Appendix}
\label{sec:appendix}
\subsection{Proof of Theorem 4.1}

The case of $k>0$ follows directly from \cite[Theorem 3]{KorpelaPG14}.
In the case $k=0$ we use a reduction from \textsc{Vertex-Cover}. In
the \textsc{Vertex-Cover} problem we are given a graph $G$, and an
integer $K$, and the task is to cover every edge of $G$ by selecting a
subset of vertices of $G$ of size at most $K$.  A reduction from
\textsc{Vertex-Cover} to the \textsc{CombCI} problem maps every edge
$e = (a,b)$ of the input graph $G$ into the vector $x_e$, where
$x_e(a) = x_e(b) = 1$, and $x_e(w) = 0$ when $w \not\in \{a, b\}$.
Furthermore, we add two vectors $x_a$ and $x_b$ which satisfy
$x_a(1)=x_b(2)=-m-1$ and which have a value of zero otherwise.
Moreover, we set $k = 0$ and $l = 1$.

For the optimal confidence area, we must allow the first element of
vector $a$ and the second element of vector $b$ be outside the
confidence area, resulting to a lower bound $x_l(j)=0$ for all
$j$. Thus, it suffices to consider a variant of the problem where the
input vectors $x_i$ are constrained to reside in $\{0,1\}^m$ and to
seek an upper bound $x_u \in \{0,1\}^m$.  To minimize the size of the
confidence area we need to minimize the number of $1$s in $x_u$.  We
consider a decision variant of the \textsc{CombCI} problem where we
must decide if there exists an $x_u$ with $\sum_j x_u(j) \leq K$ for a
given integer $K$.  An optimal vertex cover is obtained simply by
selecting the vertices $j$ with $x_u(j) = 1$ in the optimal upper
bound $x_u$.

\subsection{Proof of Theorem 4.2}

To prove Theorem 4.2 we use a linear programming relaxation
of an {\em integer linear program} (ILP) corresponding to the $k=0$
variant of the one-sided \textsc{CombCI} which we define as follows.
Let $X$ be a matrix representing the data vectors $x_1,\ldots, x_n$.
For the moment, consider the $j$th position of every vector $x_i$ in
$X$, and let this be {\em sorted in decreasing order} of $x_i(j)$.
Let $i^\prime$ denote the vector that follows vector $i$ in this
sorted order.  Moreover, at every position $j$ we only consider
vectors that are {\em strictly above} the lower bound $x_l(j)$.

The ILP we define uses binary variables $q$ to express whether a given
vector is inside the one-sided confidence area.  We have $q_i(j) = 1$
when vector $i$ is {\em contained in the confidence area} at position
$j$, and $q_i(j) = 0$ otherwise.

To compute the size of the confidence area, we introduce a set of
coefficients, denoted by $\Delta$.  We define $\Delta_i(j) = x_i(j) -
x_{i^\prime}(j)$ as the difference between the values of vectors $i$
and $i^\prime$ at position $j$.  For the vector $i$ that appears last
in the order, i.e., there is no corresponding $i^\prime$, we let
$\Delta_i(j) = x_i(j) - x_l(j)$, where $x_l(j)$ is the value of the
lower side of the confidence interval at position $j$.  Using this, we
can write the difference between the largest and smallest value at
position $j$ as the sum $\sum_{i=1}^{n} \Delta_i(j)$, and the size of
the envelope that contains all of $X$ is equal to $\sum_{j=1}^m
\sum_{i=1}^n \Delta_i(j)$.

Now, given a feasible assignment of values to $q$, i.e., an assignment
that satisfies the constraints that we will define shortly, we can
compute the size of the corresponding one-sided confidence area using
the sum $\sum_{i=1}^n \sum_{j=1}^m q_i(j)\Delta_i(j)$.  This is the
objective function of the ILP.

We continue with discussing the constraints.  Recall that $i^\prime$
denotes the vector that follows vector $i$ when the vectors are sorted
in decreasing order of their value at position $j$.  First, since
every vector $i$ may violate the confidence area at most at $l$
positions, we must make sure that $\sum_{j=1}^m q_i(j) \geq m-l$ holds
for every $i$.  Moreover, observe that if the vector at position $i$
is within the upper bound at position $j$, i.e., we have $q_i(j) = 1$,
the vector $i^\prime$ must be inside the upper bound at position $j$
as well, since $x_i(j) \geq x_{i^\prime}(j)$.  However, if the vector
$i^\prime$ is below the upper bound, this does not imply anything for
vector $i$.  Therefore, we must have $q_i(j) \leq q_{i^\prime}(j)$ for
all $i$ and $j$.

The resulting ILP is as follows:
\begin{eqnarray}
\min \; \sum_{i,j} q_i(j)\Delta_i(j) && \hbox{st.} \\
\label{eq:qsumconstraint}
\sum_{j=1}^m q_i(j) &\geq& m-l \; \; \forall i \\
\label{eq:qmonotonicity}
q_{i^\prime}(j) - q_{i}(j) &\geq& 0 \; \; \forall i, j \\
q_i(j) &\in& \{0,1\} \; \; \forall i, j.
\end{eqnarray}
There are at most $nm$ variables and $nm + n$ constraints in total in
addition to the binary constraints.  A relaxed version of this ILP is
otherwise equivalent, but allows all variables to take fractional
values in the interval $[0,1]$.  We use standard techniques to design
an approximation algorithm for the one-sided \textsc{CombCI} problem
(with $k=0$) as follows.  The algorithm first computes (in polynomial
time using a suitable LP solver) the optimal fractional solution,
denoted $q^*$, and then rounds this to an integer solution using a
simple threshold value.  We must select the threshold so that the
resulting integer solution is guaranteed to be feasible for the
original ILP.  We define the rounding procedure as follows:
\begin{equation}
\label{eq:qbar}
    \bar q_i(j) = 
\begin{cases}
    1& \text{if } q_i(j)^* \geq \frac{1}{l+1},\\ 0 & \text{otherwise}.
\end{cases}
\end{equation}
The threshold value we use is thus $1/(l+1)$.  To prove that the
algorithm is correct, we must show (i) that the rounding procedure
results in a solution where every vector $i$ is outside the confidence
area at most $l$ times, and (ii) that there are no ``holes'' in the
solution, meaning that if $\bar q_i(j) = 1$, then $\bar
q_{i^\prime}(j)$ must be equal to $1$ as well.  We start with this
latter property.
\begin{Lemma}
\label{lemma:unimodularq}
Consider position $j$, and let all vectors be sorted in decreasing
order of $x_i(j)$.  This order results in a monotonically increasing
sequence of $q_i(j)^*$ values.
\end{Lemma}
\begin{proof}
The proof is a simple observation that since Equation
\eqref{eq:qmonotonicity} holds for any feasible solution, we must have
$q_{i^\prime}(j)^* \geq q_i(j)^*$ for every $i$ and $j$.
\end{proof}

This means that if the rounding procedure of Equation \eqref{eq:qbar}
sets $\bar q_i(j) = 1$, we must have $\bar q_{i^\prime}(j) = 1$ as
well.  The result is also guaranteed to include every vector $i$
within the confidence area at least $m-l$ times.
\begin{Lemma}
\label{lemma:roundingok}
The rounding scheme of Equation \eqref{eq:qbar} gives a solution $\bar
q$ that satisfies the constraint of Equation \eqref{eq:qsumconstraint}
for all~$i$.
\end{Lemma}
\begin{proof}
We must find a worst-case distribution of values in $q_i^*$ that
requires a very small threshold value so that the rounded variables
satisfy the constraint.  Suppose that there are $m-l-1$ ones in
$q_i^*$.  Since the total sum of all elements in $q_i^*$ is at least
$m-l$, the remaining $l+1$ items must sum to $1$.  The worst-case
situation happens when all of the remaining $l+1$ items are equal to
$1/(l+1)$.  By selecting this as the threshold value, we are
guaranteed to satisfy the constraint.
\end{proof}

Next, we consider the approximation ratio of the solution given by
$\bar q$.  The proof follows standard approaches.
\begin{Lemma}
\label{lemma:approx}
The cost of the rounded solution $\bar q$ is at most $(l+1)$ times the
optimal solution to the original ILP.
\end{Lemma}
\begin{proof}
By the definition of $\bar q$ in Equation \eqref{eq:qbar}, we must
have $\bar q_i(j) \leq (l+1) q_i(j)^*$.  This implies that
\[
\sum_{i,j} \bar q_i(j) \Delta_i(j) \leq (l+1) \sum_{i,j} q_i(j)^* \Delta_i(j).
\]
The cost of $\bar q$ is thus at most $(l+1)$ times the cost of the
optimal fractional LP, which lower bounds the optimal cost of the ILP.
\end{proof}

Now, using Lemmas \ref{lemma:unimodularq}, \ref{lemma:roundingok},
and~\ref{lemma:approx} together, we have an algorithm that implies
Theorem 4.2.

\renewcommand{\figurename}{Algorithm}
\renewcommand\thefigure{\thesection.\arabic{figure}}    
\setcounter{figure}{0} 

\begin{figure}[h]
\caption{
Greedy algorithm, top-down.
\label{algo:greedy}}
\begin{algorithmic}[1]
\STATE{{\bf input:} dataset $X\in\mathbb{R}^{N\times M}$, integers $K$,$L$}
\STATE{{\bf output:} $Z\in \{0,1\}^{N\times M}$ indicating included points} 
\STATE{$I\leftarrow \{1,\ldots,N\}$
\STATE{$Z\leftarrow$zeros$(N,M)$}}
\FOR{$k=1\ldots K$}
\STATE\label{ln:fe}{$R\leftarrow$FindEnvelope$(X, I, L)$}
\STATE{$I\leftarrow I\setminus\{$ExcludeObservation$(X,R)\}$}
\ENDFOR
\STATE\label{ln:felast}{$R\leftarrow$FindEnvelope$(X, I, L)$}
\STATE{$Z(i,j)\leftarrow 1$  for all points $(i,j)$ included in $R$}
\RETURN{$Z$}
\end{algorithmic}
\end{figure}

\subsection{Greedy algorithm}

The idea of the greedy algorithm, presented in
Algorithm~\ref{algo:greedy}, is to start from the whole data envelope
and sequentially find $k$ vectors to exclude by selecting at each
iteration the vector whose removal reduces the envelope the largest
amount. %
In order to
find the envelope wrt.\ the relaxed condition that allows $l$
positions from each vector to be outside, at each iteration the set of
included data points needs to be (re)computed (line \ref{ln:fe}).  This is
realized in function FindEnvelope (details presented in Algorithm \ref{algo:envelope}), which essentially
implements a greedy algorithm solving the \textsc{CombCI} problem for
$k=0$.  After the envelope has been computed, the data vector whose
exclusion yields a maximal decrease in the size of the confidence area
is excluded. The function ExcludeObservation (details presented in Algorithm \ref{algo:rmo})  
is a variant of as the greedy MWE algorithm from
\cite{KorpelaPG14} with $k=1$, and with the ordering structure $R$ provided.
After $k$ data vectors have been excluded FindEnvelope is
called for the final set of data vectors in line \ref{ln:felast} and matrix indicating included points is returned.

\begin{figure}
\caption{
FindEnvelope$(X, I, L)$
\label{algo:envelope}}
\begin{algorithmic}[1]
\STATE{{\bf input:} dataset $X\in\mathbb{R}^{N\times M}$, 
$I\subseteq 2^N$, integer $L$}
\STATE{{\bf output:} ordering structure $R$}
\STATE{rmd$\leftarrow$zeros($N$)}
\STATE\label{line:foo}{$U\leftarrow$list()}
\STATE\label{line:orderX}{$R\leftarrow$ordering structure for observations in $X[I,:]$}
\FOR{$j=1\ldots M$}\label{line:forstart}
    \FOR{$b=0\ldots1$}
          \STATE\label{line:v}{$v\leftarrow  |X[$1st$(R,j,b),j]-X[$2nd$(R,j,b),j]|$}
          \STATE{$U$.add(key=$v$, col=$j$, up=$b$)}
   \ENDFOR
\ENDFOR\label{line:forend}
\STATE\label{line:prio}{$G\leftarrow$priorityQueue$(U)$}
\WHILE{$G \ne\emptyset$ and $\sum_i$rmd$(i)<N\cdot L$}\label{line:start_while}
   \STATE \label{line:getmax}{$(v, j, b)\leftarrow$getMaximumElement($G$)}
   \STATE{$i\leftarrow$1st$(R, j, b)$}
    \IF{rmd$(i)<L$}
      \STATE\label{line:removeij}{$R\leftarrow$remove1st$(R,j,b)$}
      \STATE\label{line:updatermd}{rmd($i$)$\leftarrow$rmd($i$)+1}
      \STATE\label{line:updatevalue1}{$v\leftarrow |X[$1st$(R,j,b),j]-X[$2nd$(R,j,b),j]|$}
      \STATE\label{line:updatevalue2}{$G$.add(key=$v$, col=$j$, up=$b$)}
     \ENDIF
\ENDWHILE\label{line:end_while}
\RETURN{$R$}
\end{algorithmic}
\end{figure}

The function FindEnvelope, presented in
Algorithm~\ref{algo:envelope}, identifies which $l$ points from each
data vector to leave outside of the confidence area in order to obtain
a confidence area of minimal size.  To efficiently maintain
information about the ordering of the values in columns an ordering
structure $R$ is used (line \ref{line:orderX}).  Each element of $R$
is a doubly-linked list $R_j$  storing the ordering information
for column $j$, with the first element corresponding to the row index
of the smallest element in $X[:, j]$.  Several functions are
associated with $R$. The function 1st$(R,j,b)$ returns the row index
of the smallest ($b=0$) or the largest ($b=1$) value for column $j$ in
$X$, and, similarly, function 2nd$(R,j,b)$ returns the row index of
the second smallest or the second largest value.  Furthermore, function
remove1st$(R,j,b)$ removes the first ($b=0$) or the last element ($b=1)$
of $R_j$.  Lines \ref{line:forstart}--\ref{line:prio} initialize a
priority queue $G$ maintaining a list of candidate points for
exclusion along with the respective gains, i.e., reductions of
confidence area.

The main part of the search for the points to exclude from the relaxed
confidence area is the while-loop in lines
\ref{line:start_while}--\ref{line:end_while}. At each iteration the
point $x^\star$ leading to maximal decrease in the size of the
confidence area is excluded, if less than $l$ elements have already
been excluded from the data vector in question.  The loop terminates
when there are no positions left in which the point with an extreme
value can be excluded without breaking the vector-wise constraints.

Our approach in the function ExcludeObservation, presented in
Algorithm~\ref{algo:rmo}, is similar to that in \cite[Algorithm
  1]{KorpelaPG14}, i.e., the index of observation to remove is the one
that results in the maximal decrease in the envelope size wrt.\
the current one. 

\begin{figure}[t!]
\caption{ExcludeObservation$(X,R)$
\label{algo:rmo}}
\begin{algorithmic}[1]
\STATE{{\bf input:} dataset $X\in\mathbb{R}^{N\times M}$, ordering structure $R$}
\STATE{{\bf output:} index of observation to exclude $i$}
\STATE{Create hash table $\Delta U$ s.t. a query with previously unused key returns value of zero}
\STATE{ $A\leftarrow\emptyset$}
\FOR{$j=1\ldots M$}
    \FOR{$b=0\ldots1$}
         \STATE{$\Delta U[$1st$(R,j,b)] \leftarrow\Delta U[$1st$(R,j,b)]$}
         \STATE{\ \ \ \ \ \ \ \ \ \ $+|X[$1st$(R,j,b),j]-X[$2nd$(R,j,b),j]|$}
         \STATE{$A\leftarrow A\cup \{$1st$(R,j,b)\}$}
    \ENDFOR
\ENDFOR
\RETURN{arg max$_{a\in A}\Delta U[a]$}
\end{algorithmic}
\end{figure}

Notice that for $l=0$ it is computationally more
efficient to use the greedy algorithm from \cite{KorpelaPG14}, as in
that case there is no need to update the set of data vectors included
in the envelope for each iteration over $k$ as done in
Algorithm~\ref{algo:greedy}.

The greedy algorithm performs well in practice, although it does not
provide approximation guarantees. Consider, e.g., a
setting with $n=5$, $m=2$, $k=0$, and $l=1$ and a data matrix $X$ such
that
\[ X^T = \left( \begin{array}{ccccc}
0 & 2\epsilon & 4\epsilon & 5\epsilon & 7\epsilon \\
3\epsilon & 2\epsilon & 0 & 1 & 1-2\epsilon
\end{array} \right)\]
 with $\epsilon$ arbitrary small. 
  The optimal solution is given by 
\[ Z_\mathsf{opt}^T = \left( \begin{array}{ccccc}
0 & 0 & 1 & 1 & 1 \\
1 & 1 & 0 & 0  & 0
\end{array} \right)\]
and has a cost $4\epsilon$,
while the greedy algorithm gives solution 
\[ Z_\mathsf{alg}^T = \left( \begin{array}{ccccc}
0 & 0 & 1 & 1 &  0\\
1 & 1 & 0 & 0  & 1
\end{array} \right)\]
with a cost of $1-3\epsilon$.

\end{document}